\documentclass[11pt,final]{scrartcl}
\usepackage{etoolbox}

\usepackage[utf8]{inputenc}

\usepackage[final]{microtype}
\usepackage[dvipsnames]{xcolor}
\usepackage{float,xspace,amssymb,amsbsy,amsmath,amsthm,mathtools,algorithm,tikz,booktabs,subcaption,thmtools}
\usepackage{complexity}
\usetikzlibrary{automata,positioning,arrows,shapes,calc,matrix,math,backgrounds,quotes,graphs,external}
\usepackage[indLines=true]{algpseudocodex}
\usepackage[final,hidelinks,bookmarks,bookmarksopen,bookmarksdepth=2]{hyperref}
\usepackage[capitalize,nameinlink]{cleveref}
\usepackage[inline]{enumitem}
\makeatletter
\let\c@author\relax
\makeatother

\usepackage[style=alphabetic,citetracker=true,sorting=nyt,sortcites,mincitenames=1,maxcitenames=2,minalphanames=3,maxalphanames=3,maxbibnames=99,backend=biber]{biblatex}

\allowdisplaybreaks
\usepackage{amsthm,authblk}
\newcommand{\email}[1]{\href{mailto:#1}{\texttt{#1}}}
\theoremstyle{definition}
\newtheorem{definition}{Definition}

\makeatletter
\renewcommand\paragraph{\@startsection{paragraph}{4}{\z@}%
                     {-12\p@ \@plus -4\p@ \@minus -4\p@}%
                     {-0.5em \@plus -0.22em \@minus -0.1em}%
                     {\normalfont\normalsize\bfseries}}
\makeatother

\usepackage[final,inline,nomarginclue,nomargin,multiuser]{fixme}
\fxsetup{theme=color,mode=multiuser}

\numberwithin{equation}{section}
\renewcommand{\part}{\mathsf{part}}
\newcommand{\ddir}{\ensuremath{d_{\rightarrow}}}
\newcommand{\TS}{\textsc{Token Sliding}\xspace}

\newcommand{\Gund}{G^{\text{und}}}

\definecolor{bostonuniversityred}{rgb}{0.8, 0.0, 0.0}
\FXRegisterAuthor{dah}{DAH}{\color{bostonuniversityred}DAH}
\FXRegisterAuthor{nb}{NB}{\color{RedOrange}NB}
\FXRegisterAuthor{ce}{CE}{\color{blue}CE}

\tikzstyle{node}=[draw,circle,minimum size=2em]
\tikzstyle{edge}=[draw,thick,->]

\usepackage[appendix=inline, bibliography=common]{apxproof}

\declaretheorem{theorem}
\declaretheorem{lemma}

\declaretheorem{corollary}

\crefname{construction}{Construction}{Construction}
\crefname{claim}{Claim}{Claim}

\usepackage{adjustbox}

\title{Directed Token Sliding\thanks{This work began during the visit of N.~Banerjee and C.~Engels to D.A.~Hoang at the Vietnam Institute for Advanced Study in Mathematics (VIASM). We would like to thank VIASM for providing a productive research environment and excellent working conditions.}}

    \author{Niranka Banerjee}
    \affil{Research Institute of Mathematical Sciences, Kyoto University, Japan\footnote{This work was supported by JSPS KAKENHI Grant Number JP20H05967.}\protect\\ \email{niranka@gmail.com}}
    \author{Christian Engels}
    \affil{National Institute of Informatics, Tokyo, Japan\footnote{This work was supported by JSPS KAKENHI Grant Number JP18H05291.}\protect\\ \email{christian.engels@gmail.com}}
    \author{Duc A.\ Hoang}
    \affil{VNU University of Science, Vietnam National University, Hanoi, Vietnam
	\protect\\ \email{hoanganhduc@hus.edu.vn}}
    \date{}

\begin{document}
\maketitle
\setcounter{footnote}{0} 

\begin{abstract}
	Reconfiguration problems involve determining whether two given configurations can be transformed into each other under specific rules.
  The \textsc{Token Sliding} problem asks whether, given two different set of tokens on vertices of a graph $G$, we can transform one into
  the other by sliding tokens step-by-step along edges of $G$ such that each resulting set of tokens forms an independent set in $G$.
  Recently, Ito et al. [MFCS 2022] introduced a directed variant of this problem. They showed that for general \emph{oriented} graphs (i.e., graphs where no pair of vertices can have directed edges in both directions), the problem remains $\mathsf{PSPACE}$-complete, and is
  solvable in polynomial time on oriented trees.

	In this paper, we further investigate the \textsc{Token Sliding} problem on various oriented graph classes. We show that the problem
  remains $\mathsf{PSPACE}$-complete for oriented split graphs, bipartite graphs and bounded treewidth graphs. Additionally,
  we present polynomial-time algorithms for solving the problem on oriented cycles and cographs.

\end{abstract}

\section{Introduction}\label{sec:intro}

\subsection{Reconfiguration Problems}\label{sec:reconf}
For the last two decades, \textit{reconfiguration problems} have emerged in different research areas, including recreational mathematics,
computational geometry, constraint satisfaction, distributed algorithms, motion planning, rerouting networks, algorithmic game theory, and
even quantum complexity theory~\cite{Heuvel13,Nishimura18,MynhardtN19,BousquetMNS24}.
In a reconfiguration setting, two \textit{feasible solutions} $S$ and $T$ of a computational problem (e.g., \textsc{Satisfiability}, \textsc{Independent Set}, \textsc{Dominating Set}, \textsc{Vertex-Coloring}, etc.) are given along with a \textit{reconfiguration rule} that
describes an \textit{adjacency relation} between feasible solutions.
The question is whether there is a sequence of adjacent feasible solutions that transforms $S$ into $T$ or vice versa.
Such a sequence, if exists, is called a \textit{reconfiguration sequence}.

One of the most well-studied reconfiguration problems is \textsc{Independent Set Reconfiguration (ISR)}.
As its name suggests, in a reconfiguration setting for \textsc{ISR}, each feasible solution corresponds to a set of tokens placed on
vertices of $G$ where no vertex has more than one token and the (vertices containing) tokens form an \textit{independent set} (i.e., a
vertex-subset of pairwise non-adjacent vertices) of $G$.
Three well-investigated reconfiguration rules are Token Sliding (TS, which involves sliding a token from one vertex to an adjacent
unoccupied vertex), Token Jumping (TJ, which involves moving a token from one vertex to any unoccupied vertex), and Token Addition/Removal
(TAR($k$), which involves adding or removing a token while maintaining at least $k$ tokens).
\textsc{ISR} under TS, introduced by \textcite{HearnD05}, is also known as the \TS problem and was the first reconfiguration
problem for independent sets studied in the literature.
We refer readers to the recent survey~\cite{BousquetMNS24} and the references therein for more details on the developments of \textsc{ISR}
and related problems.

It is worth mentioning that in almost all reconfiguration problems considered so far, reconfiguration is \textit{symmetric}: if there is a
reconfiguration sequence that transforms $S$ into $T$, then by reversing this sequence, one can obtain a reconfiguration sequence that
transforms $T$ into $S$.
From a graph-theoretic perspective, if we construct the so-called \textit{reconfiguration graph} --- a graph whose nodes are feasible
solutions of a computational problem and edges are defined by the given reconfiguration rule, then by \textit{symmetric} reconfiguration we
mean the reconfiguration graph is \textit{undirected}.
On the other hand, to the best of our knowledge, \textit{non-symmetric} reconfiguration (i.e., the reconfiguration graph is directed) has
only been studied in a few contexts, such as ``reconfiguration of vertex-colorings''~\cite{FelsnerHS09}, ``reconfiguration of independent
sets''~\cite{ItoIKNOTW22}, and recently ``token digraphs''~\cite{FernandesLPSTZ24}.

\subsection{Our Problems and Results}\label{sec:results}
In particular, \textcite{ItoIKNOTW22} proved that \TS is \PSPACE-complete on oriented graphs (i.e., directed graphs having no
symmetric pair of directed edges), \NP-complete on directed acyclic graphs, and can be solved in polynomial time on oriented trees.
In this setting, a vertex-subset of an oriented graph $G$ is \textit{independent} if it is also independent in the corresponding undirected
version of $G$ where all edge-directions are removed.
This paper is a follow-up of~\cite{ItoIKNOTW22}.
Here, we consider the computational complexity of \TS on some particular classes of oriented graphs.

In this paper, we show that \TS remains \PSPACE-complete even on oriented split graphs, oriented bipartite graphs
and oriented bounded treewidth graphs by reducing from their corresponding undirected variants, which are all known to be
\PSPACE-complete~\cite{LokshtanovM19,HearnD05,BelmonteKLMOS21,Wrochna18}.
(\cref{sec:hardness}).
On the positive side, we design a polynomial-time algorithm to solve \TS on oriented cycles and oriented cographs
(\cref{sec:poly}).

\section{Preliminaries}\label{sec:prelim}

For any undefined concepts and notations, we refer readers to~\cite{Diestel2017} and the paper~\cite{ItoIKNOTW22}.
For an oriented graph $G$, the \textit{underlying undirected graph} of $G$, denoted by $\Gund$, is the graph obtained from $G$ by removing all edge directions.
Basically, a vertex-subset of $G$ is \textit{independent} if it is independent on $\Gund$.
Similarly, the \textit{neighbourhoods} of a vertex $v$ of $G$, denoted by $N_G(v)$, is the set of all vertices adjacent to $v$ in $\Gund$, and $N_G[v] = N_G(v) \cup \{v\}$.
We sometimes omit the subscript when the graph under consideration is clear from the context.

In an instance $(G, S, T)$ of \TS, two independent sets $S$ and $T$ of an oriented graph $G$ are given.
The question is to decide whether there is a (ordered) sequence $\langle S = I_0, I_1, \dots, I_\ell = T \rangle$ from $S$ to $T$ such that
each $I_i$ ($0 \leq i \leq \ell$) is independent and $I_{i+1}$ ($0 \leq i \leq \ell-1$) is obtained from $I_i$ by sliding a token from a
vertex $u$ to one of its unoccupied out-neighbor $v$ along the \textit{directed edge} (or \textit{arc}) $(u, v)$ of $G$, where $I_i \setminus I_{i+1} = \{u\}$ and $I_{i+1} \setminus I_i = \{v\}$, i.e., $I_{i+1}$ is obtained from $I_i$ by the (valid) token-slide $u \to v$.
Thus, one can also write the above reconfiguration sequence as $\langle x_0 \to y_0, x_1 \to y_1, \dots, x_{\ell-1} \to y_{\ell - 1} \rangle$ to indicate that for $0 \leq i \leq \ell-1$, the set $I_{i+1}$ is obtained from $I_i$ by the token-slide $x_i \to y_i$.
In an undirected variant, everything can be defined similarly, except that the arc $(u, v)$ is now replaced by the undirected edge $uv$.

We conclude this section with the following simple remark: since \textcite{ItoIKNOTW22} already showed that \TS is in \PSPACE{} on
oriented graphs, to show the \PSPACE-completeness of \TS on a class of oriented graphs, it suffices to design a polynomial-time reduction
from a known \PSPACE-hard problem.

\section{Hardness Results}\label{sec:hardness}

\subsection{Oriented Split Graphs}\label{sec:split}

This section is devoted to proving the following theorem.
\begin{theorem}
	\TS is \PSPACE-complete on oriented split graphs.
\end{theorem}

We reduce from \TS on undirected split graphs, which is known to be \PSPACE-complete~\cite{BelmonteKLMOS21}.

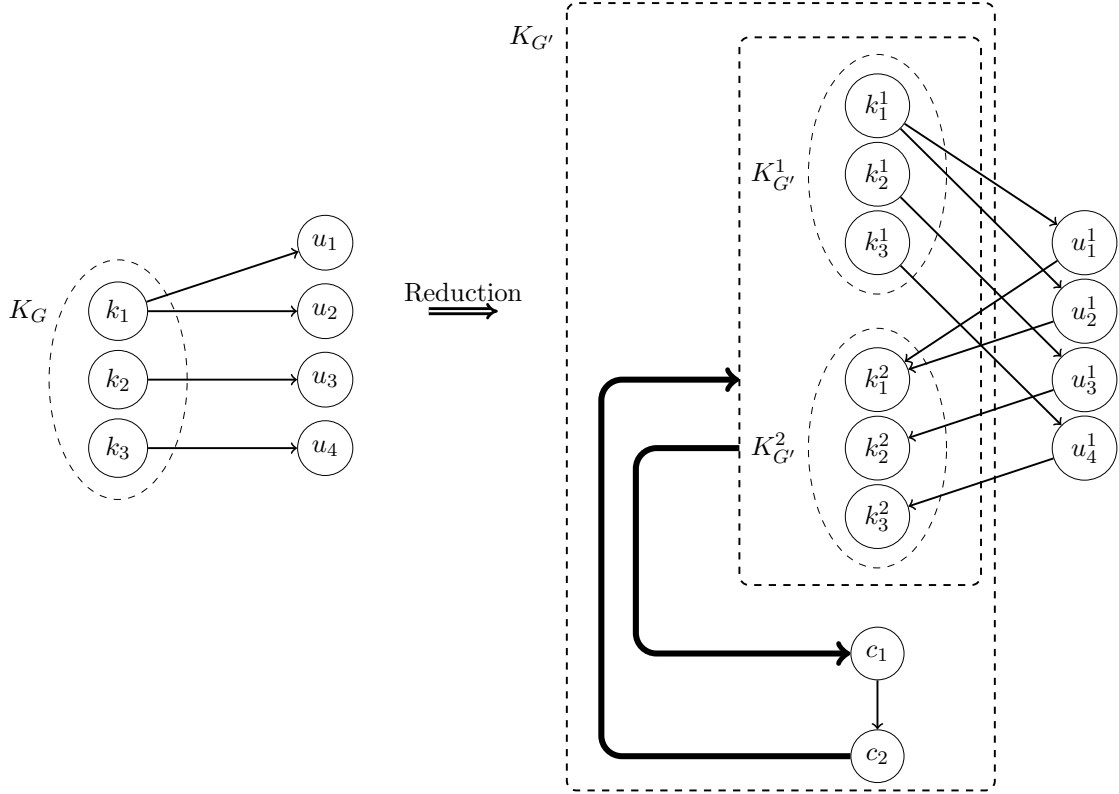
\begin{figure}[ht]
	\begin{adjustbox}{max width=\textwidth}
		\begin{tikzpicture}
			\begin{scope}[shift={(0,-3)}]
				\node[node] (k11) at (0,0) {$k_1$};
				\node[node] (k12) at (0,-1) {$k_2$};
				\node[node] (k13) at (0,-2) {$k_3$};
				\draw[dashed] (0,-1) ellipse (1cm and 1.75cm);
				\node at (-1.3,0) {$K_G$};
				\node[node] (u1) at (3,1) {$u_1$};
				\node[node] (u2) at (3,0) {$u_2$};
				\node[node] (u3) at (3,-1) {$u_3$};
				\node[node] (u4) at (3,-2) {$u_4$};
				\draw[->,thick] (k11) -- (u1);
				\draw[->,thick] (k11) -- (u2);
				\draw[->,thick] (k12) -- (u3);
				\draw[->,thick] (k13) -- (u4);
			\end{scope}

			\draw[very thick,-implies,double equal sign distance] (4.5,-3) -- node[above] {Reduction} (5.5,-3);

			\begin{scope}[shift={(11,0)}]
				\node[node] (k11) at (0,0) {$k_1^1$};
				\node[node] (k12) at (0,-1) {$k_2^1$};
				\node[node] (k13) at (0,-2) {$k_3^1$};
				\draw[dashed] (0,-1) ellipse (1cm and 1.75cm);
				\node at (-1.5,-1) {$K_{G'}^1$};

				\node[node] (k21) at (0,-4) {$k_1^2$};
				\node[node] (k22) at (0,-5) {$k_2^2$};
				\node[node] (k23) at (0,-6) {$k_3^2$};
				\draw[dashed] (0,-5) ellipse (1cm and 1.75cm);
				\node at (-1.5,-5) {$K_{G'}^2$};

				\node[node] (u1) at (3,-2) {$u_1^1$};
				\node[node] (u2) at (3,-3) {$u_2^1$};
				\node[node] (u3) at (3,-4) {$u_3^1$};
				\node[node] (u4) at (3,-5) {$u_4^1$};

				\draw[->,thick] (k11) -- (u1);
				\draw[->,thick] (k11) -- (u2);
				\draw[->,thick] (k12) -- (u3);
				\draw[->,thick] (k13) -- (u4);
				\draw[<-,thick] (k21) -- (u1);
				\draw[<-,thick] (k21) -- (u2);
				\draw[<-,thick] (k22) -- (u3);
				\draw[<-,thick] (k23) -- (u4);

				\draw[thick,dashed,rounded corners] (-2,1) rectangle (1.5,-7);
				\node[node] (c1) at (0,-8) {$c_1$};
				\node[node] (c2) at (0,-9.5) {$c_2$};
				\draw[->,thick] (c1) -- (c2);

				\draw[line width=2.5pt, <-,rounded corners=3mm] (c1) -- +(-3.5,0) -- (-3.5,-5) -- (-2,-5);
				\draw[line width=2.5pt, ->,rounded corners=3mm] (c2) -- +(-4,0) -- (-4,-4) -- (-2,-4);

				\draw[thick,dashed,rounded corners] (-4.5,1.5) rectangle (1.7,-10);
				\node at (-5,1) {$K_{G^\prime}$};
			\end{scope}
		\end{tikzpicture}
	\end{adjustbox}
	\caption{Illustration of our split graph reduction. The two large bold thick arrows indicate that there are arcs from all vertices of
		$K^1_{G^\prime} \cup K^2_{G^\prime}$ to $c_1$ and from $c_2$ to all vertices of $K^1_{G^\prime} \cup K^2_{G^\prime}$\label{fig:split}}
\end{figure}

\subsubsection*{Description of Our Reduction}
Let $(G,S,T)$ be an instance of \TS on an undirected split graph $G=(K_G \cup I_G, E)$, where the disjoint vertex-subsets $K_G$ and $I_G$
respectively induce a clique and an independent set of $G$.

We construct an instance $(G^\prime, S^\prime, T^\prime)$ of \TS on an oriented split graph
$G^\prime = (K_{G^\prime} \cup I_{G^\prime}, E^\prime)$ where the
disjoint vertex-subsets $K_{G^\prime}$ and $I_{G^\prime}$ respectively induce a clique and an independent set of $G^\prime$.
The graph $G^\prime$ is constructed as follows. (See \cref{fig:split}.)

\begin{itemize}
	\item Vertices (we ignore potential edges in this part):
	      \begin{itemize}
		      \item We add two disjoint copies of $K_G$ to $G^\prime$ and name them $K^1_{G^\prime}$ and $K^2_{G^\prime}$.
		            For convenience, the copies of $v \in K_G$ in $K^1_{G^\prime}$ and $K^2_{G^\prime}$ are respectively called $v^1$ and
					$v^2$.
		      \item We add one copy of $I_G$ to $G^\prime$ and name it $I_{G^\prime}$.
		            For convenience, the copy of $v \in I_G$ in $I_{G^\prime}$ is also called $v^1$.
		      \item We add two new vertices $c_1, c_2$ to $G^\prime$ and set $K_{G^\prime} = K^1_{G^\prime} \cup K^2_{G^\prime} \cup \{c_1, c_2\}$.
		      \item $V(G^\prime) = K_{G^\prime} \cup I_{G^\prime}$.
	      \end{itemize}

	\item Edges:
	      \begin{itemize}
		      \item We add the edge $(c_1, c_2)$ to $G^\prime$.
		      \item For each $v \in K^1_{G^\prime} \cup K^2_{G^\prime}$, we add the edges $(v, c_1)$ and $(c_2, v)$.
		      \item For each edge $uv \in E$ where $u, v \in K_G$, the orientations of the edges $u^1v^1$, $u^1v^2$, $u^2v^1$, and $u^2v^2$
			  in $G^\prime$ can be selected arbitrarily.
		      \item For each edge $uv \in E$ where $u \in K_G$ and $v \in I_G$, we add the edges $(u^1, v^1)$ and $(v^1, u^2)$ to $G^\prime$.
	      \end{itemize}
\end{itemize}

For each independent set $S$ of $G$, we define its corresponding independent set $S^\prime$ of $G^\prime$ as follows: for each $v \in S$,
add $v^1$ to $S^\prime$.
By definition of $G^\prime$, $S^\prime$ is indeed independent.
The rest of the graph forms a clique by construction.

\subsubsection*{Correctness of Our Reduction}
The above construction clearly can be done in polynomial time.
\cref{lem:split:left,lem:split:right} show that $(G, S, T)$ is a yes-instance if and only if $(G^\prime, S^\prime, T^\prime)$ is a yes-instance.

\begin{lemma}\label{lem:split:right}
	If $(G,S,T)$ is a yes-instance, then so is $(G^\prime, S^\prime, T^\prime)$.
\end{lemma}
\begin{proof}
	Let $\mathcal{I} = \langle S = I_0, I_1, \dots, I_p = T \rangle$ be a reconfiguration sequence between $S$ and $T$ in $G$.
	We shall construct a reconfiguration sequence $\mathcal{I}$ between $S^\prime$ and $T^\prime$ in $G^\prime$.
	As $\mathcal{I}$ is a reconfiguration sequence, for each $i \in \{0, 1, \dots, p-1\}$, there exist $x_i, y_i \in V$ such that
	$I_i \setminus I_{i+1} = \{x_i\}$, $I_{i+1} \setminus I_i = \{y_i\}$, and $x_i y_i \in E$.
	For each $i \in \{0, \dots, p-1\}$, we add to $\mathcal{I}$ a subsequence corresponding to $\langle I_i, I_{i+1} \rangle$ and
	join them together as follows.
	We will require the following invariant: Each subsequence will involve moving exactly one token from a vertex in
	$K^1_{G^\prime} \cup I_{G^\prime}$ to a
	vertex in $K^1_{G^\prime} \cup I_{G^\prime}$, possibly going through some vertices in $K^2_{G^\prime} \cup \{c_1, c_2\}$ if necessary.
	It will be
	clear that we always fulfill it.
	\begin{itemize}
		\item If both $x_i, y_i \in K_G$, we add the subsequence $\langle I_i^\prime, I_{i_1}^\star, I_{i_2}^\star, I_{i+1}^\prime \rangle$
		      where
		      $I_{i_1}^\star, I_{i_2}^\star, I_{i+1}^\prime$ are respectively obtained from their predecessors by the token-slides
		      $x_i^1 \to c_1$, $c_1 \to c_2$, and $c_2 \to y_i^1$.
		\item If $x_i \in K_G$ and $y_i \in I_G$, we simply add the subsequence $\langle I_i^\prime, I_{i+1}^\prime \rangle$.
		      That is, $I_{i+1}^\prime$ is obtained from $I_i^\prime$ by the token-slide $x_i^1 \to y_i^1$.
		\item If $x_i \in I_G$ and $y_i \in K_G$, we note that it is impossible to directly slide the token on $x_i^1$ to $y_i^1$, because
		      by the construction of
		      $G^\prime$, only the edge $(y^1_i, x^1_i)$ is in $E^\prime$.
		      On the other hand, also by the construction of $G^\prime$, $(x^1_i, y^2_i) \in E^\prime$.
		      Thus, we add the subsequence $\langle I_i^\prime, I_{i_1}^\star, I_{i_2}^\star, I_{i_3}^\star, I_{i+1}^\prime \rangle$ where
		      $I_{i_1}^\star, I_{i_2}^\star, I_{i_3}^\star, I_{i+1}^\prime$ are respectively obtained from their predecessors by the
		      token-slides $x^1_i \to y^2_i$,
		      $y^2_i \to c_1$, $c_1 \to c_2$, $c_2 \to y^1_i$.
	\end{itemize}
	Observe that there can be at most one token in $K_{G^\prime}$ and both $c_1$ and $c_2$ are not adjacent to any vertex in $I_{G^\prime}$.
	Combining with the definitions of $I^\prime_i, I^\prime_{i+1}$ and the assumption that sliding a token from $x_i$ to $y_i$ in $G$
	always result an
	independent set, one can verify that each intermediate token-set $I^\star_{i_j}$ in the described subsequences is independent in
	$G^\prime$.
	Naturally, two subsequences are joined at their common end (which is either $I^\prime_i$ or $I^\prime_{i+1}$).
	Joining all these subsequences gives us the desired reconfiguration sequence $\mathcal{I}$.
\end{proof}

\begin{lemma}\label{lem:split:left}
	If $(G^\prime,S^\prime,T^\prime)$ is a yes-instance, then so is $(G, S, T)$.
\end{lemma}
\begin{proof}
	Let $\mathcal{I} = \langle S^\prime = J_0, J_1, \dots, J_q = T^\prime \rangle$ be a reconfiguration sequence between $S^\prime$
	and $T^\prime$ in $G^\prime$ with our previous definition of $S',T'$.
	Let $\langle J_{i_0}, J_{i_1}, \dots, J_{i_p} \rangle$ be the maximum-length subsequence extracted from $\mathcal{I}$ such that $i_0 = 0$, $i_p = q$,
	$1 \leq i_1 < i_2 < \dots < i_{p-1} \leq q-1$, and $J_{i_j} \subseteq (K^1_{G^\prime} \cup I_{G^\prime})$ for $1 \leq j \leq p-1$.
	By definition, each $J_{i_j}$ ($1 \leq j \leq p-1$) has a corresponding independent set in $G$, say $I_j$, which can be obtained as follows: for each
	$u^1 \in J_{i_j}$, add $u$ to $I_j$.
	Our desired reconfiguration sequence $\mathcal{I}$ is constructed by joining reconfiguration sequences in $G$ between $I_j$ and $I_{j+1}$ as follows.

	If $J_{i_j}$ and $J_{i_{j+1}}$ ($1 \leq j \leq p-2$) are consecutive in $\mathcal{I}$, for the corresponding independent sets $I_j$ and $I_{j+1}$ in
	$G$, it follows that $J_{i_{j+1}}$ is obtained from $J_{i_j}$ by a token-slide between two vertices in $K^1_{G^\prime} \cup I_{G^\prime}$.
	Thus, $I_{j+1}$ can be obtained from $I_j$ in $G$ by the corresponding token-slide in $G$.
	To complete our proof, it remains to derive the same conclusion for the case that $J_{i_j}$ and $J_{i_{j+1}}$ ($1 \leq j \leq p-2$) are
	not consecutive in $\mathcal{I}$.

	Note that once a token $t$ is placed on a vertex $u^2 \in K^2_{G^\prime}$ by some token-slide in $\mathcal{I}$, one must keep moving $t$ until it is
	placed on some vertex $v^1 \in K^1_{G^\prime}$. There are two reasons.
	First, since $t$ is on a vertex of $K_{G^\prime}$, one cannot move any other token, because moving any other token $t^\prime$ must involve sliding $t^\prime$
	from a vertex in $I_{G^\prime}$ to a vertex in $K_{G^\prime}$, which cannot be done while $t$ is still in $K_{G^\prime}$.
	Second, since there is no edge directed from any vertex in $K^2_{G^\prime} \cup \{c_1, c_2\}$ to a vertex in $I_{G^\prime}$, it is impossible to directly
	slide $t$ to any vertex in $I_{G^\prime}$ without going through a vertex in $K^1_{G^\prime}$ first.

	As a result, if $J_{i_j}$ and $J_{i_{j+1}}$ ($1 \leq j \leq p-2$) are not consecutive in $\mathcal{I}$, the following scenario must
	happen: exactly
	one token $t$ that is originally in $J_{i_j}$ may be slid from some vertex $u^1 \in K^1_{G^\prime} \cup I_{G^\prime}$ to some vertex
	$v^2 \in K^2_{G^\prime}$
	and then moved around among vertices in $K_{G^\prime}$ but finally must arrive at a vertex $w^1 \in K^1_{G^\prime} \cap J_{i_{j+1}}$.
	Moreover, while $t$ is moving in $G^\prime$, no other token can move.
	Therefore, $I_{j+1}$ can be obtained from $I_j$ by iteratively applying the token-slides $u \to v$ and $v \to w$ in $G$.
	The former token-slide is possible because a token from $u^1$ can be moved to $v^2$ in $G^\prime$.
	The latter is possible because both $v$ and $w$ are in $K_G$.
	(Here, it does not matter how the corresponding token is moved in $G^\prime$ from $v^2$ to $w^1$.)
	Additionally, we note that the former token-slide is necessary only when $u^1 \in I_{G^\prime}$; otherwise, both $u$ and $w$ are in
	$K_G$ and one single
	token-slide $u \to w$ in $G$ is enough.
\end{proof}

\subsection{Oriented Bipartite Graphs}\label{sec:bipartite}

This section is devoted to proving the following theorem.
\begin{theorem}\label{thm:bipartite}
	\TS is \PSPACE-complete on oriented bipartite graphs.
\end{theorem}

A result by \textcite{LokshtanovM19} shows that \TS is \PSPACE-complete on undirected bipartite graphs.
We extend this result to oriented bipartite graphs.

\subsection*{Description of Our Reduction}

Let $(G, S, T)$ be a \TS instance on an undirected bipartite graph $G = (V, E)$.
We describe how to construct an instance $(G', S', T')$ on oriented bipartite graphs $G' = (V', E')$.

For simplicity, we will reuse the vertex names from $V$ to construct $V'$.
More precisely, we proceed as follows.
(See \cref{fig:bipartite}.)
\begin{enumerate}
	\item For each vertex $v\in V$, copy $v$ to $V^\prime$ and a new corresponding vertex $v'$.
	\item For each edge $uv \in E(G)$, add one of the directed edges $(u, v)$ or $(v, u)$ to $E'$.
	\item For each directed edge $(u, v) \in E'$ where $u, v \in V$, complete the directed $4$-cycle $(u,v)$,
	      $(v,u')$, $(u',v')$, $(v',u)$ in $E'$.
\end{enumerate}

\begin{figure}[ht]
	\centering
	\begin{adjustbox}{max width=\textwidth}
		\begin{tikzpicture}
			\node[node] (u) at (0,1) {$u$};
			\node[node] (v) at (2,1) {$v$};
			\node[node] (w) at (4,1) {$w$};
			\draw[edge,-] (u) -- (v);
			\draw[edge,-] (v) -- (w);

			\draw[very thick,-implies,double equal sign distance] (5.5,1) -- node[above] {Reduction}(6.5,1);
			\begin{scope}[shift={(8,0)}]
				\node[node] (u) at (0,2) {$u$};
				\node[node] (u') at (0,0) {$u'$};
				\node[node] (v) at (2,2) {$v$};
				\node[node] (v') at (2,0) {$v'$};
				\node[node] (w) at (4,2) {$w$};
				\node[node] (w') at (4,0) {$w'$};

				\draw[edge] (u) -- (v);
				\draw[edge] (v) -- (u');
				\draw[edge] (u') -- (v');
				\draw[edge] (v') -- (u);
				\draw[edge] (v) -- (w);
				\draw[edge] (w) -- (v');
				\draw[edge] (v') -- (w');
				\draw[edge] (w') -- (v);
			\end{scope}
		\end{tikzpicture}
	\end{adjustbox}
	\caption{Bipartite Gadget\label{fig:bipartite}}
\end{figure}
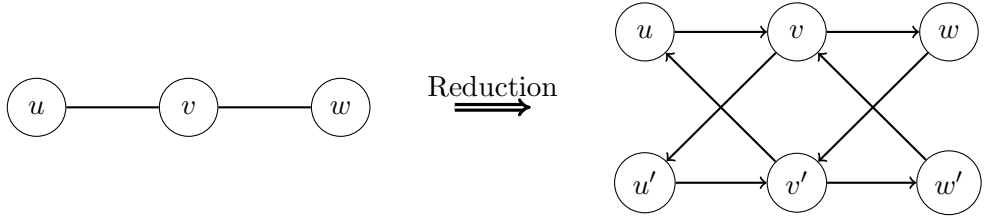

Naturally, we set $S' = S$ and $T' = T$.
Clearly, this construction can be done in  polynomial time.
One can verify that the constructed graph $G' = (V', E')$ is an oriented bipartite graph and both $S'$ and $T'$ are independent in $G'$.

\subsection*{Correctness of Our Reduction}

Before showing that our reduction is correct, we prove the following simple observation.

\begin{lemma}\label{clm:planar:neighbour}
	$N_{G'}(v)=N_{G'}(v')$ for all $v\in V'$.
\end{lemma}
\begin{proof}
	If there is an edge $(u,v)$ then we know by the construction there is an edge $(u,v')$. As we do not add edges to $v'$ unless
	preceded by an edge to $v$, the neighborhoods are equal.
\end{proof}

We are now ready to show that $(G, S, T)$ is a yes-instance if and only if $(G', S', T')$ is a yes-instance.

\begin{lemma}\label{lem:bipartite:left}
	If $(G, S, T)$ is a yes-instance, then so is $(G', S', T')$.
\end{lemma}
\begin{proof}
	Let $\mathcal{I}$ be a reconfiguration sequence that transforms $S$ into $T$ in $G$.
	We describe how to construct a reconfiguration sequence $\mathcal{I}'$ that transforms $S'$ into $T'$ in $G'$.
	For each step $u \to v$ in $\mathcal{I}$, we construct a corresponding sequence $S'_{uv}$ of token-slides in $\mathcal{I}'$ as follows: If $(u, v) \in E'$, we set $S'_{uv} = \langle u \to v \rangle$ to $\mathcal{I}'$, and if $(v, u) \in E'$, we set
	$S'_{uv} = \langle u \to v', v' \to u', u' \to v \rangle$.
	The sequence $\mathcal{I}'$ is obtained by iteratively concatenating these $S'_{uv}$ sequences.

	To see that $\mathcal{I}'$ is a reconfiguration sequence in $G'$, it suffices to verify that the token-slides in $S'_{uv}$ are valid in $G'$, i.e., each of them results an independent set.
	Let $I$ and $J$ be two independent sets of $G$ where $J$ is obtained from $I$ by the token-slide $u \to v$.
	Thus, $N_G[v]$ contains exactly one token on $u$.
	We claim that $S'_{uv}$ transforms $I$ into $J$ in $G'$.
	If $(u, v) \in E'$, since $N_G[v]$ contains exactly one token on $u$, our construction implies that $N_{G'}[v]$ contains exactly one token on $u$ and thus $S'_{uv} = \langle u \to v \rangle$ is our desired reconfiguration sequence.
	If $(v, u) \in E'$, by \cref{clm:planar:neighbour}, we have $N_{G'}[v']$ contains exactly one token on $u$, so the token-slide $u \to v'$ is valid.
	At this point, note that $N_{G'}[u]$ contains exactly one token on $v'$, and by \cref{clm:planar:neighbour}, so does $N_{G'}[u']$.
	Therefore, the token-slide $v' \to u'$ is valid. Similarly, the token-slide $u' \to v$ is valid. Our proof is complete.
\end{proof}

\begin{lemma}\label{lem:bipartite:right}
	If $(G', S', T')$ is a yes-instance, then so is $(G, S, T)$.
\end{lemma}

\begin{proof}
	Let $\mathcal{I}'$ be a reconfiguration sequence that transforms $S'$ into $T'$ in $G'$.
	We describe how to construct a reconfiguration sequence $\mathcal{I}$ that transforms $S$ into $T$ in $G$.
	For each step $x \to y$ in $\mathcal{I}'$, we construct a corresponding token-slide $S_{xy}$ of token-slides in $\mathcal{I}$ as follows.
	We set $S_{xy} = u \to v$ ($u, v \in V(G)$) if one of the following cases happens:
	\begin{itemize}
		\item $x = u \in V(G)$ and $y = v \in V(G)$.
		\item $x = u \in V(G)$ and $y = v' \in V(G') - V(G)$.
		\item $x = u' \in V(G') - V(G)$ and $y = v \in V(G)$.
		\item $x = u' \in V(G') - V(G)$ and $y = v' \in V(G') - V(G)$.
	\end{itemize}
	The sequence $\mathcal{I}$ is constructed by iteratively concatenating $S_{xy}$ and removing redundant token-slides that appear consecutively more than once.
	One can verify from our construction that $\mathcal{I}$ is indeed a reconfiguration sequence.
	Our proof is complete.
\end{proof}

\subsubsection*{A Simple Corollary}
Note that if our original graph $G$ is \textit{bounded treewidth} (i.e., its treewidth is at most some positive constant $c$), then so is our constructed graph $G'$.
Thus, the same reduction holds for oriented bounded treewidth graphs.
Additionally, \textcite{Wrochna18} proved that there exists a constant $c$ such that \TS is \PSPACE-complete on graphs with treewidth
at most $c$.\footnote{Indeed, Wrochna showed a more general result for graphs with bandwidth at most $c$. Since any graph with bandwidth at
most $c$ also has treewidth and pathwidth at most $c$, his result also applies to bounded treewidth graphs.}
Therefore, we have the following corollary.
\begin{corollary}
	There exists a constant $c$ such that \TS is \PSPACE-complete on oriented graphs having treewidth at most $c$.
\end{corollary}

\section{Polynomial-Time Results}\label{sec:poly}

\subsection{Oriented Cycles}\label{sec:cycle}
\begin{definition}\label{def:direct-dist}
	We define the \emph{directed distance}, $\ddir(u,v)$, between two vertices as
	\[
		\ddir(u,v) =  \begin{cases}
			k      & \text{if the shortest directed path from $u$ to $v$ has length $k$,} \\
			\infty & \text{otherwise.}
		\end{cases}
	\]
\end{definition}

\begin{definition}\label{def:cycle:locked}
	We call a token sliding problem $(G,S,T)$ on an oriented cycle $G$ \emph{locked} if the following holds:
	\begin{itemize}
		\item For all vertices $u,v$, $\ddir(u,v)< \infty$.
		\item For every $s\in S$ there exists an $s'\in S$ with $\ddir(s,s')=2$.
		\item $S\neq T$.
	\end{itemize}
\end{definition}

\begin{figure}
	\centering
	\begin{subfigure}{0.3\textwidth}
		\centering
		\begin{tikzpicture}[scale=0.5]
			\node[draw,circle,fill] (1) at (1*360/6: 2cm) {};
			\node[draw,circle,accepting] (2) at (2*360/6: 2cm) {};
			\node[draw,circle,fill] (3) at (3*360/6: 2cm) {};
			\node[draw,circle,accepting] (4) at (4*360/6: 2cm) {};
			\node[draw,circle,fill] (5) at (5*360/6: 2cm) {};
			\node[draw,circle,accepting] (6) at (6*360/6: 2cm) {};

			\draw[thick,->] (1) -- (2);
			\draw[thick,->] (2) -- (3);
			\draw[thick,->] (3) -- (4);
			\draw[thick,->] (4) -- (5);
			\draw[thick,->] (5) -- (6);
			\draw[thick,->] (6) -- (1);
		\end{tikzpicture}
		\caption{Locked Cycle\label{fig:locked-cycle}}
	\end{subfigure}
	\begin{subfigure}{0.3\textwidth}
		\centering
		\begin{tikzpicture}[scale=0.5]
			\node[draw,circle,fill] (1) at (1*360/6: 2cm) {};
			\node[draw,circle] (2) at (2*360/6: 2cm) {};
			\node[draw,circle] (3) at (3*360/6: 2cm) {};
			\node[draw,circle,accepting] (4) at (4*360/6: 2cm) {};
			\node[draw,circle,fill] (5) at (5*360/6: 2cm) {};
			\node[draw,circle,accepting] (6) at (6*360/6: 2cm) {};

			\draw[thick,->] (1) -- (2);
			\draw[thick,->] (3) -- (2);
			\draw[thick,->] (3) -- (4);
			\draw[thick,->] (4) -- (5);
			\draw[thick,->] (5) -- (6);
			\draw[thick,->] (6) -- (1);
		\end{tikzpicture}
		\caption{Cycle (no-instance)}
	\end{subfigure}
	\begin{subfigure}{0.3\textwidth}
		\centering
		\begin{tikzpicture}[scale=0.5]
			\node[draw,circle,fill] (1) at (1*360/6: 2cm) {};
			\node[draw,circle] (2) at (2*360/6: 2cm) {};
			\node[draw,circle] (3) at (3*360/6: 2cm) {};
			\node[draw,circle,accepting] (4) at (4*360/6: 2cm) {};
			\node[draw,circle,fill] (5) at (5*360/6: 2cm) {};
			\node[draw,circle,accepting] (6) at (6*360/6: 2cm) {};

			\draw[thick,->] (1) -- (2);
			\draw[thick,->] (2) -- (3);
			\draw[thick,->] (3) -- (4);
			\draw[thick,->] (4) -- (5);
			\draw[thick,->] (5) -- (6);
			\draw[thick,->] (1) -- (6);
		\end{tikzpicture}
		\caption{Cycle (yes-instance)}
	\end{subfigure}
	\caption{Cycles (filled vertices are in $S$ and circled vertices are in $T$)\label{fig:cycle}}
\end{figure}
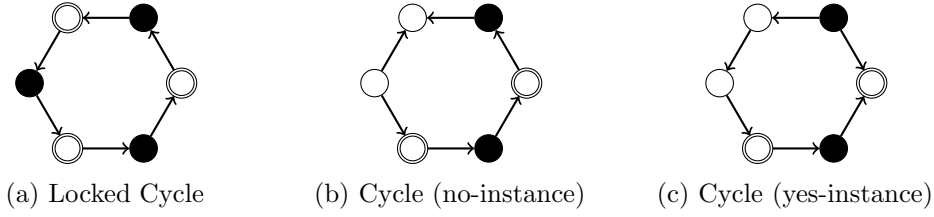

In this section, we assume that the vertices $v_1, \dots, v_n$ of an oriented cycle are always arranged in a clockwise order (e.g., see \cref{fig:cycle}).
We say that the cycle is \textit{completely clockwise} if its arcs are $(v_1, v_2), (v_2, v_3), \dots, (v_{n-1}, v_n)$, and $(v_n, v_1)$, and \textit{completely counter-clockwise} if its arcs are $ (v_n, v_{n-1}), \dots, (v_2, v_1)$, and $(v_1, v_n)$.

\begin{lemma}\label{lem:cycle:non-uniform}
	\TS on oriented cycles that are not completely clockwise or counter-clockwise is solvable in $2T(n)$ where $T(n)$ is the running time of
	solving \TS for oriented trees on $n$ vertices.
\end{lemma}
\begin{proof}
	Given $(G,S,T)$, let $G=(V,E)$ with $V=\{v_1,\dots,v_n\}$ and $(v_1,v_2),(v_1,v_n)\in E$. This has to exist as the cycle
	is not completely clockwise or counter-clockwise. Notice that no token can enter $v_1$ from $G- v_1$.

	For the following we notice that solving a reconfiguration problem on an oriented path or a collection of oriented paths is
	easy via a simple backtracking algorithm or by the result from~\cite{ItoIKNOTW22}.
	Now we perform a case distinction.

	\begin{description}
		\item[$v_1\in S\cap T$:] We remove $v_1$ from the graph and have a \TS problem on
			an oriented forest which is easy to solve.
		\item[$v_1\not\in S\cup T$:] As no token can cross from $v_2$ to $v_n$ we can again, remove $v_1$ from the \TS problem and
			arrive at an oriented forest.
		\item[$v_1\in T\setminus S$:] As no token can enter $v_1$ this problem is always unsatisfiable.
		\item[$v_1\in S\setminus T$:] Here we perform two operations and take the logical OR of their decision.
			\begin{itemize}
				\item Remove the edge $(v_1,v_n)$ from the graph. Notice that $v_n\not\in S$ as otherwise the initial configuration was
				invalid. Now we arrive again at a reconfiguration problem on an oriented path which is easy to solve.
				\item Remove the edge $(v_1,v_2)$ from the graph. Notice again that $v_2\not\in S$ as it is a valid reconfiguration initial
				state. We again arrive at a problem on an oriented path which is easy to solve.
			\end{itemize}
			Now we take the OR of these two runs, meaning accept if at least one accepts and reject otherwise. The correctness is given as there
			is only one token in $v_1$ which has to move to $G - v_1$ and has to either take $(v_1,v_2)$ or $(v_1,v_n)$.
	\end{description}

	The correctness is clear from the previous discussion.
\end{proof}

\begin{theorem}\label{lem:cycle}
	\TS is solvable in polynomial time on oriented cycles.
\end{theorem}
\begin{proof}
	It is clear that testing if a cycle is locked is easy in polynomial time. Let us assume that the cycle is not locked and
	that it has the majority of edges
	clockwise as the other case is symmetric. There are now two cases.
	\begin{description}
		\item[There exists an edge $(u,v)$ which is counter-clockwise:]
			In this case we invoke \cref{lem:cycle:non-uniform}.
		\item[All edges are clockwise:] Notice that in this case there will always be a token that can move as we are not locked.
			For every $\pi:S\rightarrow T$ we can define the \emph{total distance}:
			\[
				\Delta(\pi,S) = \sum_{s\in S} \ddir(s,\pi(s)).
			\]

			Let $s_1,\dots,s_m$ be the tokens in order such that
			for each $s_i$, $s_{i+1}$ minimizes the directed distance from $s_i$ to $s_{i+1}$. Similarly, let us order the targets
			to get $t_1,\dots,t_m$. Let us pick the mapping such that $\ddir(s_1,t_i)$ is the largest.

			Let $S'_\pi$ be the set of all tokens
			that are not on their targets and let $S''_\pi\subseteq S'_\pi$ of tokens that can move
			and still be an independent set where both sets are with respect to a base set $S_\pi$.

			Notice now that our mapping
			$\Delta(\tau,S_{\tau,i}) +1 = \Delta(\tau,S_{\tau,i+1})$.
			This is now easy to see we can move $s_1$ to $\tau(s_1)$ (by perhaps moving other tokens). Each step $s_1$ takes will reduce
			the total distance by one. Each step any other token takes will also reduce the total distance by one as every token gets
			closer to its target (as the targets are in the same order).

			As $\Delta(\tau,S)<\infty$ for all $\tau$ and all $S$, the tokens have always a valid way to move.
			This means the instance is always solvable.
	\end{description}
\end{proof}

\subsection{Oriented Cographs}\label{sec:cograph}
A \textit{cograph} is a graph that does not contain a $P_4$ (a path on four vertices) as an induced subgraph.
\Textcite{KaminskiMM12} showed that \TS on undirected cographs can be solved in linear time.
By slightly modifying their algorithm, one can achieve a linear-time algorithm for the problem on oriented cographs.
Though it is trivial, for the sake of completeness, we present the algorithm here.

Recall that a \textit{co-component} of a graph $G$ is the subgraph of $G$ induced by vertices of a connected component of the complement
graph $\overline{G} = (V, \{uv \mid u, v \in V \text{ and } uv \notin E(G)\})$.
A graph $G$ is a cograph if and only if the complement of any nontrivial connected induced subgraph of $G$ is disconnected.
\begin{theorem}
	\TS on oriented cographs can be solved in linear time.
\end{theorem}
\begin{proof}
	Let $(G, S, T)$ be an instance of \TS on an oriented cograph $G$.
	We describe a linear-time algorithm to decide whether it is a yes-instance.

	If $|S| \neq |T|$ then clearly it is a no-instance.
	Otherwise, we consider the following cases:
	\begin{itemize}
		\item If $|S| = |T| = 1$, then we check if there is a directed path from $s$ to $t$ in $G$ where $S = \{s\}$ and $T = \{t\}$. If so, it is a yes-instance (we can slide the token on $s$ to $t$ along the directed path). Otherwise, it is a no-instance.
		\item If $|S| = |T| \geq 2$, we consider whether the underlying undirected graph $\Gund$ of $G$ is connected.
		      If not, we solve for each oriented component of $\Gund$ separately and combine the results.
		      Otherwise, note that an independent set of $\Gund$ must belong to a unique co-component.
		      Thus, if $S$ and $T$ belong to two different co-components, it is a no-instance.
		      Otherwise, we recursively solve the problem on the co-component containing $S$ and $T$.
	\end{itemize}
	The running time and correctness of this algorithm are clear from~\cite{KaminskiMM12}.
	Our proof is complete.
\end{proof}

\clearpage
\printbibliography
\end{document}